\documentclass[preprint,12pt]{article}
\usepackage{bbm}
\usepackage{amsfonts}
\usepackage{mathrsfs}
\usepackage{amssymb}
\usepackage{amsmath}
\usepackage{amsthm}
\usepackage{color}
\usepackage{url}
\usepackage{hyperref}
\usepackage{multirow}
\usepackage{bbding}
\usepackage[ruled,vlined]{algorithm2e}
\usepackage{latexsym}
\usepackage[final]{graphicx}
\usepackage{cases}
\usepackage{amsthm}

\usepackage{lastpage}
\usepackage{epsfig}
\usepackage[misc]{ifsym}
\usepackage{lineno}

\newtheorem{definition}{Definition}
\newtheorem{theorem}{Theorem}
\newtheorem{lemma}{Lemma}
\newtheorem{remark}{Remark}

\begin{document}
\begin{center}

\textbf{\Large{Novel CRT-based Asymptotically Ideal Disjunctive Hierarchical Secret Sharing Scheme}}

\vspace{0.7cm}

Hongju Li$^{1,2}$, Jian Ding$^{1,3}$ $\mbox{(\Letter)}$, Fuyou Miao$^{2}$, Cheng Wang$^{3}$, Cheng Shu$^{3}$

\vspace{3mm}
\footnotesize{
$^1$School of Mathematics and Big Data, Chaohu University, Hefei 238024, China\\
$^2$School of Computer Science and Technology, University of Science and Technology of China, Hefei 230026, China\\
$^3$School of Artificial Intelligence and Big Data, Hefei University, Hefei 230000, China}

\footnotetext{\footnotesize {This research was supported by Research Project of Chaohu University under Grants No.
KYQD-202220, No. hxkt20250173, No. hxkt20250174, and University Natural Science Research Project of Anhui Province under Grant No. 2024AH051324.}}

\footnotetext{\footnotesize {\noindent$\mbox{\Letter}$ Corresponding author.}}

\footnotetext{\footnotesize {E-mail addresses: 058072@chu.edu.cn (H. Li), dingjian\_happy@163.com (J. Ding), mfy@ustc.edu.cn (F. Miao)}, wangcheng@stu.hfuu.edu.cn (C. Wang), shucheng@stu.hfuu.edu.cn (C. Shu)}
\end{center}

\noindent\textbf{Abstract}: \emph{Disjunctive Hierarchical Secret Sharing (DHSS)} scheme is a type of secret sharing scheme in which the set of all participants is partitioned into disjoint subsets, and each subset is said to be a level with different degrees of trust and different thresholds. In this work, we focus on the Chinese Remainder Theorem (CRT)-based DHSS schemes due to their ability to accommodate flexible share sizes. We point out that the ideal DHSS scheme of Yang et al. (ISIT, 2024) and the asymptotically ideal DHSS scheme of Tiplea et al. (IET Information Security, 2021) are insecure. Consequently, existing CRT-based DHSS schemes either exhibit security flaws or have an information rate less than $\frac{1}{2}$. To address these limitations, we propose a CRT-based asymptotically perfect DHSS scheme that supports flexible share sizes. Notably, our scheme is asymptotically ideal when all shares are equal in size. Its information rate achieves one and it has computational security.


\section {Introduction}
Secret sharing (SS) is a cryptographic technique used to distribute a secret among a group of participants \cite{Shamir1979,Blakley1979}. It consists of two fundamental phases: a share generation phase and a secret reconstruction phase. In the former, for any secret from the secret space, the dealer divides the secret into multiple shares, and each participant receives a share from the dealer. In the latter, any authorized subset of participants can reconstruct the secret with their shares, and any given unauthorized subset of participants cannot do so. An SS scheme is said to be \emph{perfect} if each unauthorized subset of participants learns no information about the secret. The \emph{information rate} is a crucial efficiency metric in SS schemes. It is defined as the ratio of the size of the secret space to the size of the maximum share space. An SS scheme is said to be \emph{ideal} if it is perfect and its information rate achieves one. It is well known that the maximum information rate of perfect SS schemes is one.

A $(t,n)$-threshold scheme is an SS scheme in which any $t$ or more participants can reconstruct the secret, but any group of less than $t$ participants learns no information about the secret. This implies that the $(t,n)$-threshold scheme establishes equal trust among all participants in the secret reconstruction phase. But in practical implementations, the shares of participants are often correlated with their organizational roles or positions. \emph{Disjunctive Hierarchical Secret Sharing (DHSS)} schemes are particularly well-suited for this application scenario \cite{Simmons}. In a DHSS scheme, the set of all participants is partitioned into $m$ disjoint subsets $\mathcal{P}_1,\mathcal{P}_2,\ldots,\mathcal{P}_m$, referred to as levels. Each level has a distinct degree of trust and a corresponding threshold $t_{\ell}$ for $\ell\in \{1,2,\ldots,m\}$. The secret can be reconstructed if and only if there is a number $\ell\in \{1,2,\ldots,m\}$ such that the total number of participants from the first $\ell$ levels meets or exceeds $t_{\ell}$. When a level $\mathcal{P}_\ell$ contributes $r_{\ell}$ participants for $r_{\ell}<t_{\ell}$, the remaining $(t_{\ell}-r_{\ell})$ participants can be drawn from higher levels $\mathcal{P}_1,\mathcal{P}_2,\ldots,\mathcal{P}_{\ell-1}$.

DHSS schemes have been constructed from a variety of mathematical tools, including geometry \cite{Simmons}, vector spaces \cite{Brickell1989}, Birkhoff interpolation \cite{Tassa2007}, polymatroids \cite{Chenqi2022}, the Chinese Remainder Theorem (CRT) for integer ring \cite{Harn-Miao2014,Oguzhan-Ersoy2016,Tiplea2021} and the CRT for polynomial ring \cite{Yangjing2024}. Among these methods, CRT-based DHSS schemes exhibit a natural advantage: an inherent flexibility in assigning shares of varying sizes to participants. This capability makes them particularly well-suited for hierarchical access structures. Therefore, we focus on CRT-based DHSS schemes in this work. In 2014, Harn et al. \cite{Harn-Miao2014} use the CRT for integer ring to construct a DHSS scheme for the first time. However, Ersoy et al. \cite{Oguzhan-Ersoy2016} demonstrate that the scheme of Harn et al. \cite{Harn-Miao2014} is insecure. They further propose a DHSS scheme based on the CRT for the integer ring and hash functions. This scheme is \emph{asymptotically perfect} and computationally secure, as it publishes many hash values. Its information rate is less than $\frac{1}{2}$. Tiplea et al. \cite{Tiplea2021} also use the CRT for integer ring to construct an asymptotically perfect DHSS scheme. Their scheme achieves an information rate approaching one, making it \emph{asymptotically ideal}. This asymptotic perfectness stems directly from the use of CRT for integer ring. Recently, Yang et al. \cite{Yangjing2024} construct an ideal DHSS scheme using the CRT for the polynomial ring for the first time.

\emph{Our contributions}. We point out that the ideal DHSS scheme of Yang et al. \cite{Yangjing2024} and the asymptotically ideal scheme of Tiplea et al. \cite{Tiplea2021} are insecure. We further propose an asymptotically ideal DHSS scheme with small share size. More specifically, we make the following contributions:

1) For clarity, we present the attack method for the 2-level case of Yang et al. scheme \cite{Yangjing2024}. The same approach can also be applied to break the DHSS scheme of Tiplea et al. \cite{Tiplea2021}. Consequently, existing CRT-based DHSS schemes either exhibit security flaws or have an information rate less than $\frac{1}{2}$.

2) We propose a novel asymptotically ideal DHSS scheme based on the CRT for polynomial ring and one-way hash functions. Unlike the unconditionally secure ideal schemes in  \cite{Brickell1989,Tassa2007,Chenqi2022}, our scheme achieves flexible and smaller share sizes (see Table 1), although it provides computational security and requires publishing more values. Compared with the DHSS schemes in \cite{Harn-Miao2014, Oguzhan-Ersoy2016,Tiplea2021,Yangjing2024}, our scheme is secure and asymptotically ideal at the same time.

\begin{table}[!htb]\scriptsize
  \centering
   {\bf Table 1.}  Disjunctive hierarchical secret sharing schemes, where $t_m$ is the scheme's biggest threshold, and $|\mathcal{P}_{\ell}|$ is the cardinality of the level $\mathcal{P}_{\ell}$.\\
  \begin{tabular}{ccccccc}
  \hline
    \multirow{2}*{Schemes}  & \multirow{2}*{Security}    & Accommodate             &\multirow{2}*{Perfectness}  &Information              &\multirow{2}*{Share size}\\
    ~                        &~                        &flexible share sizes      &~                            &rate $\rho$                 &~\\
  \hline
  \hline
  \cite{Brickell1989}           &\multirow{2}*{unconditional}   &\multirow{2}*{No}  &\multirow{2}*{Yes}          &\multirow{2}*{$\rho=1$}  &{$mt_m^2\log_2 p$ bits,}\\
     (Scheme 2)                 &~                              &~                   &~                         &~                      &$p>\max_{\ell\in\{1,2,\ldots,m\}}\{|\mathcal{P}_{\ell}|\}$\\
  \hline
  \cite{Tassa2007}(when $t_m$     &\multirow{2}*{unconditional}         &\multirow{2}*{No} &\multirow{2}*{Yes}      &\multirow{2}*{$\rho=1$} &{$\log_2 p$ bits,}\\
     was small.)                  &~                         &~                        &~    &~   &$p>\begin{pmatrix} n+1 \\ t_m\end{pmatrix}\frac{(t_m-2)(t_m-1)}{2}+t_m$\\
  \hline
  \multirow{2}*{\cite{Chenqi2022}}    &\multirow{2}*{unconditional}  &\multirow{2}*{No}   &\multirow{2}*{Yes}           &\multirow{2}*{$\rho=1$}
                                                                                    &{$\frac{1}{2}\sum_{\ell=1}^{m-1}t_{\ell}(t_{\ell}-1)\log_2 p$ bits,}\\
   ~                                  &~                   &~                         &~                            &~  &$p>\max_{\ell\in\{1,2,\ldots,m\}}\{|\mathcal{P}_{\ell}|\}$\\
  \hline
  \cite{Harn-Miao2014}          &No                       &Yes                     &No                        &$\rho<1$          &*\\
  \hline
   \cite{Tiplea2021}            &No                       &Yes                     &No                         &$\rho<1$         &*\\
  \hline
  \cite{Yangjing2024}           &No                       &Yes                     &No                        &$\rho=1$          &*\\
  \hline
  \cite{Oguzhan-Ersoy2016}      &computational                      &Yes                  &Asymptotic                   &$\rho<\frac{1}{2}$   &*\\
  \hline
  Our scheme                    &computational                      &Yes                   &Asymptotic                &$\rho=1$               &$d_0\log_2 p$ bits, $d_0\geq 1, p>n$\\
  \hline
  \end{tabular}
 \end{table}

\emph{Paper organization}. After some preliminaries in Section \ref{Sec: Prelim}, the results of this work are organized as follows. In Section \ref{Sec: Security analysis of Yang}, we analyze the security of the scheme of Yang et al. \cite{Yangjing2024}. In Section \ref{Sec: our scheme}, we present a novel asymptotically ideal DHSS scheme, accompanied by its security analysis. We conclude in Section \ref{Sec: conclusion}.

\section{Preliminaries}\label{Sec: Prelim}
In this section, we will introduce basic notions about secret sharing, Chinese Reminder Theorem (CRT) for polynomial ring and the one-way hash function.
\subsection{Secret sharing}
Let $n, N_1, N_2$ be positive integers such that $N_1<N_2$. Denote by $[n]=\{1,2,\ldots,n\}$ and $[N_1, N_2]=\{N_1,N_1+1,\ldots,N_2\}$. Let $\mathbf{X}$ and $\mathbf{Y}$ be random variables, and let $\mathsf{H}(\mathbf{X})$ be the Shannon entropy of $\mathbf{X}$. Denote by $\mathsf{H}(\mathbf{X}|\mathbf{Y})$ the conditional Shannon entropy of $\mathbf{X}$ given $\mathbf{Y}$.

\begin{definition}[Secret sharing scheme]
Let $\mathcal{P}=\{P_1,P_2,\ldots,P_n\}$ be a group of $n$ participants. A secret sharing scheme of $n$ participants consists of a share generation phase and a secret reconstruction phase, as shown below.
\begin{itemize}
  \item[1)] Share Generation Phase. For any secret from the secret space $\mathcal{S}$, the dealer applies the map
                      \[\mathsf{SHARE}\colon\mathcal{S}\times\mathcal{R}\mapsto\mathcal{S}_1\times\mathcal{S}_2\times\dotsb\times\mathcal{S}_n\]
    to assign shares to participants from $\mathcal{P}$, where $\mathcal{S}_i$ is the share space of the participant $P_i$, and $\mathcal{R}$ is a set of random inputs.
  \item[2)] Secret Reconstruction Phase. Any given authorized subset $\mathcal{A}\subseteq \mathcal{P}$ can reconstruct the secret by using their shares and the map
                     \[\mathsf{RECON}\colon\prod_{P_i\in\mathcal{A}}\mathcal{S}_i\mapsto\mathcal{S},\]
   while any given unauthorized subset cannot reconstruct the secret.
\end{itemize}
\end{definition}

We usually take the number $i$ as the $i$-th participant $P_i$ in this work. This means that $[n]=\mathcal{P}$. If any given unauthorized subset of participants learns nothing about the secret, the secret sharing scheme is said to be perfect.

\begin{definition}[Information rate, \cite{Ning2018}]\label{def: Information rate}
For a secret sharing scheme of $n$ participants, its information rate is defined as
                                \[\rho=\frac{\mathsf{H}(\mathbf{S})}{\max_{i\in [n]}^{}{\mathsf{H}(\mathbf{S}_i)}},\]
where $\mathbf{S}$ and $\mathbf{S}_i$ are random variables corresponding to the secret and the share of the $i$-th participant, respectively. When $\mathbf{S}$ and $\mathbf{S}_i$  are random and uniformly distributed in the secret space $\mathcal{S}$ and each share space $\mathcal{S}_i$, it holds that
                               \[\rho=\frac{\log_2|\mathcal{S}|}{\max_{i\in [n]}^{}{\log_2|\mathcal{S}_i|}},\]
where $|\mathcal{S}|$ is the number of elements in the secret space, and $|\mathcal{S}_i|$ is the number of elements in the share space of the $i$-th participant.
\end{definition}

The information rate $\rho$ of a perfect secret sharing scheme satisfies $\rho\leq 1$.  Specifically, a secret sharing scheme is said to be ideal if it is perfect and has information rate one.

\begin{definition}[Disjunctive hierarchical secret sharing scheme, \cite{Tassa2007}]\label{def: Multilevel secret sharing}
Assume that a set $\mathcal{P}$ of $n$ participants is partitioned into $m$ disjoint subsets $\mathcal{P}_1, \mathcal{P}_2,\ldots, \mathcal{P}_m$, namely,
        \[\mathcal{P}=\cup_{\ell=1}^{m} \mathcal{P}_{\ell},~and~\mathcal{P}_{\ell_1}\cap\mathcal{P}_{\ell_2}=\varnothing~ for~any~ 1\leq \ell_1<\ell_2\leq m.\]
For a threshold sequence $t_1, t_2, \ldots, t_m$ such that $1\leq t_1< t_2<\cdots<t_m\leq n$, the Disjunctive Hierarchical Secret Sharing (DHSS) scheme with the given threshold sequence is a secret sharing scheme such that the following correctness and privacy are satisfied.
\begin{itemize}
  \item[1)] Correctness. The secret can be reconstructed by any given element of $\Gamma$, where
       \[\Gamma=\{\mathcal{A} \subseteq \mathcal{P}: \exists \ell \in [m]~such~that~|\mathcal{A}\cap(\bigcup_{w=1}^{\ell}\mathcal{P}_{w})|\geq t_{\ell}\}.\]
  \item[2)] Privacy. The secret cannot be reconstructed by any given $\mathcal{B}\notin \Gamma$.
\end{itemize}
A DHSS scheme is said to be ideal if it is perfect and has an information rate one.
\end{definition}
\begin{definition}[Asymptotically ideal DHSS scheme, \cite{Quisquater2002}]\label{def: Asymptotically Ideal DHSS}
For a DHSS scheme with the secret space $\mathcal{S}$ and share spaces $\mathcal{S}_i, i\in [n]$, it is said to be asymptotically ideal if the following asymptotic perfectness and asymptotic maximum information rate are satisfied.
\begin{itemize}
  \item[1)] Asymptotic perfectness. For all $\epsilon_1>0$, there is a positive integer $\sigma_1$ such that for all $\mathcal{B}\notin\Gamma$ and $|\mathcal{S}|>\sigma_1$, the loss entropy
                       \[\Delta(|\mathcal{S}|)=\mathsf{H}(\mathbf{S})-\mathsf{H}(\mathbf{S}|\mathbf{V}_{\mathcal{B}})\leq \epsilon_1,\]
   where $\mathsf{H}(\mathbf{S})\neq 0$, and $\mathbf{S}, \mathbf{V}_{\mathcal{B}}$ are random variables corresponding to the secret and the knowledge of $\mathcal{B}$, respectively.
  \item[2)] Asymptotic maximum information rate. For all $\epsilon_2>0$, there is a positive integer $\sigma_2$ such that for all $\mathcal{B}\notin\Gamma$ and $|\mathcal{S}|>\sigma_2$, it holds that
      \[\frac{\max_{i\in [n]}^{}{\mathsf{H}(\mathbf{S}_i)}}{\mathsf{H}(\mathbf{S})}\leq 1+\epsilon_2.\]
\end{itemize}
\end{definition}

\subsection{Chinese Reminder Theorem and one-way hash function}

\begin{lemma}[CRT for polynomial ring, \cite{Ning2018,Ding2023}]\label{le: CRT}
Let $\mathbb{F}$ be a finite field and $m_1(x),m_2(x),\ldots,\\m_n(x)\in \mathbb{F}[x]$ be pairwise coprime polynomials. Denote by $M(x)=\prod\limits_{i=1}^{n}m_i(x)$, $M_i(x)=M(x)/m_i(x)$, and $\lambda_i(x)\equiv M_i^{-1}(x)\pmod {m_i(x)}$. For any given polynomials $y_1(x),y_2(x),\\\ldots,y_n(x)\in \mathbb{F}[x]$ and a system of congruences
          \[ y(x)\equiv y_i(x) \pmod {m_i(x)},~ \mathrm{for~ all}~ i\in [n],\]
it holds that
          \[y(x)\equiv\sum\limits_{i=1}^{n}\lambda_i(x)M_i(x)y_i(x) \pmod {M(x)}.\]
If the degree of $y(x)$ satisfies $\deg(y(x))<\deg(M(x))$, the solution is unique and we denote it as
          \[y(x)=\sum\limits_{i=1}^{n}\lambda_i(x)M_i(x)y_i(x) \pmod {M(x)}.\]
\end{lemma}

\begin{definition}[One-way hash function, \cite{RM1985}] A one-way hash function is a function that satisfies the following conditions:
   \begin{itemize}
     \item[(i)] The function $h(\cdot)$ is publicly known.
     \item[(ii)] The input $x$ of the function is of arbitrary length, and the output $h(x)$ is of a fixed length.
     \item[(iii)] Given $h(\cdot)$ and $x$, computing $h(x)$ is easy.
     \item[(iv)] Given an image $y$ of the function of $h(\cdot)$, it is hard to find a message $x$ such that $h(x)=y$, and given $x$ and $h(x)$, it is hard to find another message
           $x^{\prime}\neq x$ such that $h(x^{\prime})=h(x)$.
     \end{itemize}
   \end{definition}

\section{Security analysis of Yang et al. scheme}\label{Sec: Security analysis of Yang}
In this section, we review the scheme of Yang et al. \cite{Yangjing2024}, and present an attack method to show that the ideal DHSS scheme of Yang et al. \cite{Yangjing2024} and the asymptotically ideal DHSS scheme of Tiplea et al. \cite{Tiplea2021} are insecure.

\subsection{Review of Yang et al. scheme}\label{Subsec: review of Yang scheme}

For clarity in presentation, we rewrite the DHSS scheme of Yang et al. with total levels $m=2$. Let $\mathcal{P}$ be a set of $n$ participants, and it is partitioned into $2$ disjoint subsets $\mathcal{P}_1$ and $\mathcal{P}_2$. Denote by $n_{1}=|\mathcal{P}_{1}|$ and $n_{2}=|\mathcal{P}_{2}|$, then $n=n_1+n_2$. Let $t_1, t_2$ be thresholds such that $1\leq t_1<t_2\leq n_2$ and $t_{1}\leq n_{1}$. The scheme of Yang et al. consists of a share generation phase and a secret reconstruction phase.

\textbf{1) Share Generation Phase}. Let $p$ be a prime integer, and $\mathbb{F}_p$ be a finite field with $p$ elements.
 \begin{itemize}
       \item The dealer chooses a publicly known integer $d_0\geq 1$, and sets $m_0(x)=x^{d_0}\in \mathbb{F}_p[x]$. The dealer selects publicly known pairwise coprime polynomials $m_i(x)\in \mathbb{F}_p[x], i\in [n]$ such that the following three conditions are satisfied.
          \begin{itemize}
             \item[(i)] For all $i\in [n]$, $m_0(x)$ and $m_i(x)$ are coprime.
             \item[(ii)]Denote by $d_i=\deg(m_i(x))$ for $i\in [n]$, it holds that $d_0\leq d_1\leq d_2\leq\cdots\leq d_n$.
            \item[(iii)] $d_0+\sum\limits_{i=n-t_{\ell}+2}^{n}d_i\leq \sum\limits_{i=1}^{t_{\ell}}d_i$ for all $\ell\in \{1,2\}$.
           \end{itemize}
       \item For any given secret $s(x)\in \{g(x)\in \mathbb{F}_p[x]:\deg(g(x))<d_0\}$, the dealer randomly chooses two polynomials
             \[\alpha_{\ell}(x)\in \mathcal{G}_{\ell}=\{g(x)\in \mathbb{F}_p[x]:\deg(g(x))<(\sum\limits_{i=1}^{t_{\ell}}d_i)-d_0\}, \ell\in \{1,2\}.\]
        Let $f_{\ell}(x)=s(x)+\alpha_{\ell}(x)x^{d_0}$ for $\ell\in \{1,2\}$, then the dealer computes
         \[c_i(x)=\begin{cases}
            f_1(x) \pmod {m_i(x)},\mathrm{if}~i\in [n_1],\\
            f_2(x) \pmod {m_i(x)},\mathrm{if}~i\in [n_1+1,n],\\
         \end{cases}\]
         and sends each share $c_i(x)$ to the $i$-th participant $P_i$.

        \item The dealer publishes $w_i(x)=(f_2(x)-c_i(x))\pmod{m_i(x)}$ for $i\in [n_1]$.

   \end{itemize}

\textbf{2) Secret Reconstruction Phase}. For any $\mathcal{A} \subseteq \mathcal{P}$ such that $\mathcal{A}^{(\ell)}=A\cap (\cup_{w=1}^{\ell}P_{w}), \\|\mathcal{A}^{(\ell)}|\geq t_{\ell}$ for some $\ell\in \{1,2\}$, participants of $\mathcal{A}$ pool their shares and corresponding public polynomials to determine the polynomial
                          \[f_{\ell}(x)=\sum\limits_{i\in \mathcal{A}^{(\ell)}}\lambda_{i,\mathcal{A}^{(\ell)}}(x)M_{i,\mathcal{A}^{(\ell)}}(x)c_i^{(\ell)}(x) \pmod {M_{\mathcal{A}^{(\ell)}}(x)},\]
and reconstruct the secret $s(x)=f_{\ell}(x)\pmod{m_0(x)}$, where $M_{\mathcal{A}^{(\ell)}}(x)=\prod\limits_{i\in\mathcal{A}^{(\ell)}}m_i(x)$, $M_{i,\mathcal{A}^{(\ell)}}(x)=M_{\mathcal{A}^{(\ell)}}(x)/m_i(x)$, $\lambda_{i,\mathcal{A}^{(\ell)}}(x)\equiv M_{i,\mathcal{A}^{(\ell)}}^{-1}(x)\pmod {m_i(x)}$, and
        \[c_i^{(\ell)}(x)=\begin{cases}
            c_i(x),~\mathrm{if}~\ell=1,i\in \mathcal{A}^{(\ell)}\subseteq [n_1],\\
            c_i(x)+w_i(x),~\mathrm{if}~\ell=2, i\in \mathcal{A}^{(\ell)}\cap [n_1],\\
            c_i(x),~\mathrm{if}~\ell=2, i\in \mathcal{A}^{(\ell)}\cap [n_1+1,n].\\
         \end{cases}\]

\subsection{An attack method on the Yang et al. scheme}

In the scheme of Yang et al. \cite{Yangjing2024} with total levels $m=2$, we let $\mathcal{P}$ be a set of $n=7$ participants. It is partitioned into $m=2$ disjoint subsets $\mathcal{P}_1=\{P_1,P_2,P_3\}, \mathcal{P}_2=\{P_4,P_5,P_6,P_7\}$ such that $n_{1}=|\mathcal{P}_{1}|=3$ and $n_{2}=|\mathcal{P}_{2}|=4$. Let $t_1=2, t_2=3$ be the threshold sequence, and denote by $\mathcal{B}=\{P_4,P_5\}$. Clearly, $|\mathcal{B}\cap\mathcal{P}_{1}|=0<t_{1}$ and $|\mathcal{B}\cap(\bigcup_{w=1}^{2}\mathcal{P}_{w})|=2<t_{2}$. This shows that
             \[\mathcal{B}\notin \{\mathcal{A} \subseteq \mathcal{P}: \exists \ell\in\{1,2\} ~such~that~|\mathcal{A}\cap(\bigcup_{w=1}^{\ell}\mathcal{P}_{w})|\geq t_{\ell}\}.\]
Yang et al. claimed that participants of $\mathcal{B}$ learn no information about the secret. We will prove that participants of $\mathcal{B}$ can reconstruct the secret with the following three steps.

\textbf{Step 1}. Participants of $\mathcal{B}$ obtain $f_2(x)-f_1(x)\in \mathbb{F}_p[x]$ from public polynomials.

It is easy to check that $\deg(f_1(x))<\sum\limits_{i=1}^{t_{1}}d_i$ and $\deg(f_2(x))< \sum\limits_{i=1}^{t_{2}}d_i$. Since $t_1<t_2$, then it holds that $\deg(f_2(x)-f_1(x))< \sum\limits_{i=1}^{t_{2}}d_i$. For $i\in [n_1]$, it is known that $c_i(x)=f_1(x)\pmod{m_i(x)}$. Since polynomials $w_i(x)=(f_2(x)-c_i(x))\pmod{m_i(x)}, i\in [n_1]$ are public, then participants of $\mathcal{B}$ obtain
                       \[w_i(x)=(f_2(x)-f_1(x))\pmod{m_i(x)}, i\in [n_1].\]
Since $\sum\limits_{i\in [n_1]}d_i> \deg(f_2(x))$, the polynomial $f_2(x)-f_1(x)\in \mathbb{F}_p[x]$ can be determined by using CRT for polynomial ring, namely,
                         \[f_2(x)-f_1(x)=\sum\limits_{i\in [n_1]}\lambda_{i,[n_1]}(x)M_{i,[n_1]}(x)w_i(x) \pmod {M_{[n_1]}(x)}\]
where
         \[M_{[n_1]}(x)=\prod\limits_{i\in [n_1]}m_i(x), M_{i,[n_1]}(x)=M_{[n_1]}(x)/m_i(x), \lambda_{i,[n_1]}(x)\equiv M_{i,[n_1]}^{-1}(x)\pmod {m_i(x)}.\]

\textbf{Step 2.} Participants of $\mathcal{B}$ can get $f_1(x)\in \mathbb{F}_p[x]$ from $f_2(x)-f_1(x)\in \mathbb{F}_p[x]$ and their shares.

For any $i\in \mathcal{B}$, let $u_i(x)=(f_2(x)-f_1(x))\pmod{m_i(x)}$, and then participants of $\mathcal{B}$ get
            \begin{displaymath}
             \begin{aligned}
              c_i(x)-u_i(x)&=f_2(x)-(f_2(x)-f_1(x))\pmod{m_i(x)}\\
                           &=f_1(x)\pmod{m_i(x)}.
            \end{aligned}
          \end{displaymath}
Since $\sum\limits_{i\in \mathcal{B}}d_i>\deg(f_1(x))$, then the polynomial $f_1(x)\in \mathbb{F}_p[x]$ can be determined by using CRT for polynomial ring, namely,
                         \[f_1(x)=\sum\limits_{i\in \mathcal{B}}\lambda_{i,\mathcal{B}}(x)M_{i,\mathcal{B}}(x)(c_i(x)-u_i(x)) \pmod {M_{\mathcal{B}}(x)}\]
where
         \[M_{\mathcal{B}}(x)=\prod\limits_{i\in \mathcal{B}}m_i(x), M_{i,\mathcal{B}}(x)=M_{\mathcal{B}}(x)/m_i(x), \lambda_{i,\mathcal{B}}(x)\equiv M_{i,\mathcal{B}}^{-1}(x)\pmod {m_i(x)}.\]

\textbf{Step 3.} Participants of $\mathcal{B}$ reconstruct the secret $s(x)=f_1(x)\pmod{m_0(x)}$.

\begin{remark}
    In the scheme of Yang et al. \cite{Yangjing2024} with total levels $m=2$ and $n_1\geq t_2$, the secret can be reconstructed by any $\mathcal{B}\subseteq P$ of cardinality $|\mathcal{B}|\geq t_1$. In this case, the scheme of Yang et al. is insecure, and it is not a DHSS scheme. Besides, we find that the scheme of Yang et al. is not perfect if $t_1\leq n_1<t_2$.
\end{remark}
\begin{remark}
    The same approach can be applied to break the DHSS scheme of Tiplea et al. \cite{Tiplea2021}.
\end{remark}

\section{A novel asymptotically ideal DHSS scheme}\label{Sec: our scheme}
In this section, we will construct a novel asymptotically ideal DHSS scheme by using CRT for polynomial ring and one-way hash functions. Our scheme is given in subsection \ref{Subsec:our DHSS}, and we analyze its security in subsection \ref{Subsec: security of our DHSS}.

\subsection{Our scheme}\label{Subsec:our DHSS}
Let $\mathcal{P}$ be a group of $n$ participants, and it is partitioned into $m$ disjoint subsets $\mathcal{P}_1, \mathcal{P}_2,\ldots,$ $ \mathcal{P}_m$.
Let $n_{\ell}=|\mathcal{P}_{\ell}|$ and $N_{\ell}=\sum_{w=1}^{\ell} n_{w}$ for $\ell\in [m]$. Denote by $t_1, t_2, \ldots, t_m$ a threshold sequence such that $1\leq t_1<t_2<\cdots<t_m$ and $t_{\ell}\leq n_{\ell}$ for $\ell\in [m]$. Let $\lfloor x \rfloor$ be the biggest integer not more than $x$. Our scheme consists of a share generation phase and a secret reconstruction phase.

\textbf{1) Share Generation Phase}. Let $p$ be a prime integer, and $\mathbb{F}_p$ be a finite field with $p$ elements.
 \begin{itemize}
       \item \textbf{Identities.} The dealer chooses a publicly known integer $d_0\geq 1$, and sets $m_0(x)=x^{d_0}\in \mathbb{F}_p[x]$. The dealer selects publicly known pairwise coprime polynomials $m_i(x)\in \mathbb{F}_p[x], i\in [n]$ such that the following three conditions are satisfied.
          \begin{itemize}
             \item[(i)] For all $i\in [n]$, $m_0(x)$ and $m_i(x)$ are coprime.
             \item[(ii)]Denote by $d_i=\deg(m_i(x))$ for $i\in [n]$, it holds that $d_0\leq d_1\leq d_2\leq\cdots\leq d_n$.
            \item[(iii)] $d_0+\sum\limits_{i=n-t_{\ell}+2}^{n}d_i\leq \sum\limits_{i=1}^{t_{\ell}}d_i$ for all $\ell\in [m]$.
           \end{itemize}
       \item \textbf{Shares}. Let $\mathcal{S}=\{g(x)\in \mathbb{F}_p[x]:\deg(g(x))<d_0\}$ be the secret space.
          \begin{itemize}
           \item[Step 1.] For any $s(x)\in \mathcal{S}$, the dealer randomly chooses polynomials
                 \[\alpha_{\ell}(x)\in \mathcal{G}_{\ell} =\{g(x)\in \mathbb{F}_p[x]:\deg(g(x))<(\sum\limits_{i=1}^{t_{\ell}}d_i)-d_0\}, \ell\in [m].\]
              Let $f_{\ell}(x)=s(x)+\alpha_{\ell}(x)x^{d_0}$ for $\ell\in [m]$.
           \item[Step 2.] The dealer sends each share $s_i(x)$ to the $i$-th participant, where
         \[s_i(x)=\begin{cases}
            c_i(x),\mathrm{if}~i\in [N_{m-1}],\\
            f_m(x) \pmod {m_i(x)},\mathrm{if}~i\in [N_{m-1}+1,N_m],\\
         \end{cases}\]
         and the polynomials
            \[c_{i}(x)=c_{i,0}+c_{i,1}x+\cdots+c_{i,d_i-1}x^{d_i-1}\in \mathbb{F}_p[x], i\in [N_{m-1}]\]
         are randomly chosen by the dealer.
         \end{itemize}

        \item \textbf{Hierarchy}.
          \begin{itemize}
           \item[Step 1.]The dealer selects $m$ publicly known distinct one-way hash functions $h_1(\cdot), h_2(\cdot),\\\dots, h_m(\cdot)$. These functions take a value of any length as a input, and output a value of length $\lfloor \log_{2} p \rfloor$. For any $i\in [N_{m-1}]$, denote by
               \[H_{\ell}(s_i(x))=H_{\ell}(c_i(x))=h_{\ell}(c_{i,0})+h_{\ell}(c_{i,1})x+\cdots+h_{\ell}(c_{i,d_i-1})x^{d_i-1}\in \mathbb{F}_p[x], \ell\in [m].\]
           \item[Step 2.] The dealer publishes
                  \[w_i^{(\ell)}(x)=(f_{\ell}(x)-H_{\ell}(s_i(x)))\pmod{m_i(x)}~\mathrm{for~all}~\ell\in [m-1], i\in [N_{\ell}],\]
            and $w_i^{(m)}(x)=(f_{m}(x)-H_{m}(s_i(x)))\pmod{m_i(x)}$ for all $i\in [N_{m-1}]$.
            \end{itemize}
 \end{itemize}

\textbf{2) Secret Reconstruction Phase}. For any $\mathcal{A} \subseteq \mathcal{P}$ such that $\mathcal{A}^{(\ell)}=\mathcal{A}\cap (\cup_{w=1}^{\ell}\mathcal{P}_{w}),\\
 |\mathcal{A}^{(\ell)}|\geq t_{\ell}$ for some $\ell\in [m]$, participants of $\mathcal{A}$ compute
                          \[f_{\ell}(x)=\sum\limits_{i\in \mathcal{A}^{(\ell)}}\lambda_{i,\mathcal{A}^{(\ell)}}(x)M_{i,\mathcal{A}^{(\ell)}}(x)s_i^{(\ell)}(x) \pmod {M_{\mathcal{A}^{(\ell)}}(x)},\]
and reconstruct the secret $s(x)=f_{\ell}(x)\pmod{m_0(x)}$, where $M_{\mathcal{A}^{(\ell)}}(x)=\prod\limits_{i\in\mathcal{A}^{(\ell)}}m_i(x)$, $M_{i,\mathcal{A}^{(\ell)}}(x)=M_{\mathcal{A}^{(\ell)}}(x)/m_i(x)$, $\lambda_{i,\mathcal{A}^{(\ell)}}(x)\equiv M_{i,\mathcal{A}^{(\ell)}}^{-1}(x)\pmod {m_i(x)}$, and
        \[s_i^{(\ell)}(x)=\begin{cases}
            H_{\ell}(s_i(x))+w_i^{(\ell)}(x),~\mathrm{if}~\ell\in [m-1], i\in \mathcal{A}^{(\ell)}\subseteq [N_{\ell}],\\
            H_{m}(s_i(x))+w_i^{(m)}(x),~\mathrm{if}~\ell=m, i\in \mathcal{A}^{(\ell)}\cap [N_{m-1}],\\
            s_i(x),~\mathrm{if}~\ell=m, i\in \mathcal{A}^{(\ell)}\cap [N_{m-1}+1,N_m].\\
         \end{cases}\]

\subsection{Security analysis of our scheme}\label{Subsec: security of our DHSS}
We will prove the correctness, asymptotic perfectness and asymptotic maximum information rate of our scheme in this subsection.
\begin{theorem}[Correctness]\label{Theorem:Correctness of our scheme} The secret can be reconstructed by any authorized subset
       \[\mathcal{A}\in\{\mathcal{A} \subseteq \mathcal{P}: \exists \ell \in [m]~such~that~|\mathcal{A}\cap(\bigcup_{w=1}^{\ell}\mathcal{P}_{w})|\geq t_{\ell}\}.\]
\end{theorem}
\begin{proof}
Let $\mathcal{A}\subseteq \mathcal{P}$ such that $\mathcal{A}^{(\ell)}=\mathcal{A}\cap (\cup_{w=1}^{\ell}P_{w}), |\mathcal{A}^{(\ell)}|\geq t_{\ell}$ for some $\ell\in [m]$. Without loss of generality, we assume that
          \[\mathcal{A}^{(\ell)}=\{i_{1},i_{2},\dots,i_{|\mathcal{A}^{(\ell)}|}\}, i_{1}<i_{2}<\cdots<i_{{|\mathcal{A}^{(\ell)}|}},i_{|\mathcal{A}^{(\ell)}|}\in \mathcal{P}_\ell~\mathrm{and}~|\mathcal{A}^{(\ell)}|\geq t_{\ell}.\]
We will prove that the subset $\mathcal{A}^{(\ell)}$ is an authorized subset.

By using the shares $s_i(x), i\in \mathcal{A}^{(\ell)}$, the publicly known one-way hash function $h_{\ell}(\cdot)$, the publicly known polynomials $m_i(x), i\in \mathcal{A}^{(\ell)}$, and the publicly known polynomials $w_i^{(\ell)}(x)=(f_{\ell}(x)-H_{\ell}(s_i(x)))\pmod{m_i(x)}$ for $i\in \mathcal{A}^{(\ell)}$, participants of $\mathcal{A}^{(\ell)}$ compute
         \[s_i^{(\ell)}(x)=\begin{cases}
            H_{\ell}(s_i(x))+w_i^{(\ell)}(x),~\mathrm{if}~\ell\in [m-1], i\in \mathcal{A}^{(\ell)}\subseteq [N_{\ell}],\\
            H_{m}(s_i(x))+w_i^{(m)}(x),~\mathrm{if}~\ell=m, i\in \mathcal{A}^{(\ell)}\cap [N_{m-1}],\\
            s_i(x),~\mathrm{if}~\ell=m, i\in \mathcal{A}^{(\ell)}\cap [N_{m-1}+1,N_m].\\
         \end{cases}\]
for $i\in \mathcal{A}^{(\ell)}$, and get the system of congruences
      \begin{displaymath}
         \left\{\begin{aligned}
            f_{\ell}(x)\equiv &s_{i_1}^{(\ell)}(x) \pmod{m_{i_1}(x)},\\
             f_{\ell}(x)\equiv &s_{i_2}^{(\ell)}(x) \pmod{m_{i_2}(x)},\\
                     &\vdots\\
             f_{\ell}(x)\equiv &s_{i_{|\mathcal{A}^{(\ell)}|}}^{(\ell)}(x) \pmod{m_{i_{|\mathcal{A}^{(\ell)}|}}(x)}.\\
           \end{aligned} \right.
        \end{displaymath}
Based on the CRT for polynomial ring in Lemma \ref{le: CRT}, it holds that
   \[f_{\ell}(x)\equiv\sum\limits_{i\in \mathcal{A}^{(\ell)}}\lambda_{i,\mathcal{A}^{(\ell)}}(x)M_{i,\mathcal{A}^{(\ell)}}(x)s_i^{(\ell)}(x) \pmod {M_{\mathcal{A}^{(\ell)}}(x)},\]
where $M_{\mathcal{A}^{(\ell)}}(x)=\prod\limits_{i\in\mathcal{A}^{(\ell)}}m_i(x)$, $M_{i,\mathcal{A}^{(\ell)}}(x)=M_{\mathcal{A}^{(\ell)}}(x)/m_i(x)$ and $\lambda_{i,\mathcal{A}^{(\ell)}}(x)\equiv M_{i,\mathcal{A}^{(\ell)}}^{-1}(x)\pmod {m_i(x)}$. Since $d_0\leq d_1\leq d_2\leq\cdots\leq d_n$, $s(x)\in \mathcal{S}, \alpha_{\ell}(x)\in \mathcal{G}_{\ell}$, and $f_{\ell}(x)=s(x)+\alpha_{\ell}(x)x^{d_0}$ for $\ell\in [m]$, the degree of the polynomial $f_{\ell}(x)$ satisfies
                \[\deg{(f_{\ell}(x))}<\sum\limits_{i=1}^{t_{\ell}}d_i\leq\sum\limits_{i\in \mathcal{A}^{(\ell)}}d_i,\]
where $|\mathcal{A}^{(\ell)}|\geq t_{\ell}$. From the CRT for polynomial ring in Lemma \ref{le: CRT}, it holds that
  \[f_{\ell}(x)=\sum\limits_{i\in \mathcal{A}^{(\ell)}}\lambda_{i,\mathcal{A}^{(\ell)}}(x)M_{i,\mathcal{A}^{(\ell)}}(x)s_i^{(\ell)}(x) \pmod {M_{\mathcal{A}^{(\ell)}}(x)},\]
and the secret $s(x)=f_{\ell}(x)\pmod{m_0(x)}$.
\end{proof}

Now we will prove the \emph{asymptotic perfectness} of our scheme. Recall that $f_{\ell}(x)=s(x)+\alpha_{\ell}(x)x^{d_0}$ for $\ell\in [m]$, where $s(x)\in \mathcal{S}=\{g(x)\in \mathbb{F}_p[x]:\deg(g(x))<d_0\}$ and $\alpha_{\ell}(x)\in \mathcal{G}_{\ell}=\{g(x)\in \mathbb{F}_p[x]:\deg(g(x))<(\sum\limits_{i=1}^{t_{\ell}}d_i)-d_0\}$. As a result, the polynomial $f_{\ell}(x)$ is random and uniform in $\{f_{\ell}(x)\in \mathbb{F}_p[x]:\deg(f_{\ell}(x))<\sum\limits_{i=1}^{t_{\ell}}d_i\}$ when $s(x)\in \mathcal{S}$ and $\alpha_{\ell}(x)\in \mathcal{G}_{\ell}$ are random and uniform.

Assume that the subset $\mathcal{B}\subset \mathcal{P}$ is an unauthorized subset. Therefore,
            \[\mathcal{B}\notin \{\mathcal{A} \subseteq \mathcal{P}: \exists \ell \in [m]~such~that~|\mathcal{A}\cap(\bigcup_{w=1}^{\ell}\mathcal{P}_{w})|\geq t_{\ell}\},\]
namely,
    \[\mathcal{B}\cap(\cup_{w=1}^{\ell}\mathcal{P}_{w})|<t_{\ell}~\mathrm{for~all}~\ell\in [m].\]
Participants of $\mathcal{B}$ know their shares, the upper bounds of the degrees of $f_{\ell}(x)\in \mathbb{F}_p[x], \ell\in [m]$, and all publicly known polynomials and one-way hash functions. Namely, the dealer will guess the secret by first selecting $(g_1(x),g_2(x),\ldots,g_m(x)\in (\mathbb{F}_p[x])^m$ satisfying the following five conditions, and then computing $g_m(x)\pmod{m_0(x)}$.
\begin{itemize}
  \item[(i)] For all $\ell\in [m]$, it holds that $\deg(g_{\ell}(x))<\sum\limits_{i=1}^{t_{\ell}}d_i$.
  \item[(ii)] $g_{1}(x)\equiv g_{1}(x)\equiv\cdots\equiv g_{m}(x)\pmod{m_0(x)}$.
  \item[(iii)] It holds that
            \[g_{\ell}(x)\equiv (w_i^{(\ell)}(x)+H_{\ell}(s_i(x)))\pmod{m_i(x)} ~\mathrm{for~all}~\ell\in [m-1], i\in \mathcal{B}\cap [N_{\ell}],\]
            and $g_{m}(x)\equiv (w_i^{(m)}(x)+H_{m}(s_i(x)))\pmod{m_i(x)} ~\mathrm{for~all}~i\in \mathcal{B}\cap [N_{m-1}]$.
  Namely,
             \begin{displaymath}
         \left\{\begin{aligned}
            g_{1}(x)\equiv &(w_i^{(1)}(x)+H_{1}(s_i(x)))\pmod{m_i(x)}~\mathrm{for~all}~i\in \mathcal{B}\cap [N_1],\\
            g_{2}(x)\equiv &(w_i^{(2)}(x)+H_{2}(s_i(x)))\pmod{m_i(x)}~\mathrm{for~all}~i\in \mathcal{B}\cap [N_2],\\
                     &\vdots\\
            g_{m-1}(x)\equiv &(w_i^{(m-1)}(x)+H_{m-1}(s_i(x)))\pmod{m_i(x)}~\mathrm{for~all}~i\in \mathcal{B}\cap [N_{m-1}],\\
            g_{m}(x)\equiv &(w_i^{(m)}(x)+H_{m}(s_i(x)))\pmod{m_i(x)}~\mathrm{for~all}~i\in \mathcal{B}\cap [N_{m-1}].\\
           \end{aligned} \right.
        \end{displaymath}
  \item[(iv)] For $i\in \mathcal{B}\cap [N_{m-1}+1, N_m]$, it holds that $g_{m}(x)\equiv s_i(x)\pmod{m_i(x)}$.
  \item[(v)] For any given $i\in [N_{m-1}]$ and $i\notin \mathcal{B}$, there is a level $\mathcal{P}_{\ell_1}$ such that $i\in \mathcal{P}_{\ell_1}$. There is a polynomial $\widetilde{s}_i(x)\in \mathbb{F}_p[x]$ such that $\deg{(\widetilde{s}_i(x))}<d_i$, and $H_{\ell}(\widetilde{s}_i(x))=(g_{\ell}(x)-w_i^{(\ell)}(x))\pmod{m_i(x)}$ for all $\ell\in [\ell_1,m]$.
  \end{itemize}
Note that the conditions (iii) and (v) are used to make the polynomials $g_1(x),g_2(x),\ldots,\\g_m(x)$ satisfy publicly known polynomials $w_i^{(\ell)}(x)$.

\begin{lemma}\label{le: loss entropy}
Let $\mathcal{V}_{\mathcal{B}}$ and $\mathcal{V}^\prime_{\mathcal{B}}$ be the conditions (i) to (v) and (i) to (iv), respectively. For all $\epsilon_1>0$, there is a positive integer $\sigma_1$ such that for all $|\mathcal{S}|=p^{d_0}>\sigma_1$, it has that
                               \[0<\mathsf{H}(\mathbf{S}|\mathbf{V}^\prime_{\mathcal{B}})-\mathsf{H}(\mathbf{S}|\mathbf{V}_{\mathcal{B}})<\epsilon_1,\]
where $\mathbf{V}_{\mathcal{B}}$ and $\mathbf{V}^\prime_{\mathcal{B}}$ are random variables corresponding to $\mathcal{V}_{\mathcal{B}}$ and $\mathcal{V}^\prime_{\mathcal{B}}$, respectively.
\end{lemma}
\begin{proof}
Recall that the shares $s_i(x), i\in [N_{m-1}]$ are randomly selected by the dealer and functions $h_{\ell}(\cdot), \ell\in [m]$ are distinct one-way hash functions. This shows that the condition (v) can eliminate a group of polynomials $(g_1(x),g_2(x),\ldots,g_m(x))\in (\mathbb{F}_p[x])^m$ with a negligible probability when $|\mathcal{S}|$ is big enough. Therefore, $\mathsf{H}(\mathbf{S}|\mathbf{V}^\prime_{\mathcal{B}})-\mathsf{H}(\mathbf{S}|\mathbf{V}_{\mathcal{B}})$ is negligible when $|\mathcal{S}|$ is big enough. This gives the proof.
\end{proof}

Denote by
\begin{equation}\label{Eq: definiton F}
            \mathcal{F}=\{(g_1(x),g_2(x),\ldots,g_m(x)\in (\mathbb{F}_p[x])^m: ~\mathrm{the~conditions~(i)~to~(iv)~}\mathrm{are~satisfied}.\}
       \end{equation}
For any given $s(x)\in \mathcal{S}$, we will compute how many possible $(g_1(x),g_2(x),\ldots,g_m(x))\\\in \mathcal{F}$ such that $s(x)=g_{m}(x)\pmod{m_0(x)}$ in Lemma \ref{le: preimage}. After that, we will determine the cardinality of $\mathcal{F}$ in Lemma \ref{le: cardinality of F}. Based on Lemmas \ref{le: loss entropy}, \ref{le: preimage} and \ref{le: cardinality of F}, we determine the conditional entropy $\mathsf{H}(\mathbf{S}|\mathbf{V}^\prime_{\mathcal{B}})$ and give the proof of the asymptotic perfectness of our scheme in Theorem \ref{Th: our perfectness}.

\begin{lemma}\label{le: preimage}
Define the mapping $\Phi$ by
 \[\Phi: \mathcal{F}\mapsto \mathcal{S}, (g_1(x),g_2(x),\ldots,g_m(x))\mapsto g_m(x)\pmod{m_0(x)}.\]\
 For any $s(x)\in \mathcal{S}$, let
           \[\Phi^{-1}(s(x))=\{(g_1(x),g_2(x),\ldots,g_m(x))\in \mathcal{F}: g_{m}(x)\equiv s(x)\pmod{m_0(x)}\},\]
then the cardinality of the set $\Phi^{-1}(s(x))$ is $|\Phi^{-1}(s(x))|=p^{\theta}$, where
            \[\theta=\sum_{\ell=1}^{m}\left(\sum\limits_{i=1}^{t_\ell}d_i-\sum\limits_{i\in \mathcal{B}\cap [N_{\ell}]}d_i-d_0\right)\geq 0.\]
\end{lemma}
\begin{proof}
For any $s(x)\in \mathcal{S}$, and $(g_1(x),g_2(x),\ldots,g_m(x))\in \Phi^{-1}(s(x))$, it holds that
   \begin{displaymath}
         \left\{\begin{aligned}
            g_{1}(x)\equiv &g_{1}(x)\equiv\cdots\equiv g_{m}(x)\equiv s(x)\pmod{m_0(x)},\\
            g_{1}(x)\equiv &(w_i^{(1)}(x)+H_{1}(s_i(x)))\pmod{m_i(x)}~\mathrm{for~all}~i\in \mathcal{B}\cap [N_1],\\
            g_{2}(x)\equiv &(w_i^{(2)}(x)+H_{2}(s_i(x)))\pmod{m_i(x)}~\mathrm{for~all}~i\in \mathcal{B}\cap [N_2],\\
                     &\vdots\\
            g_{m-1}(x)\equiv &(w_i^{(m-1)}(x)+H_{m-1}(s_i(x)))\pmod{m_i(x)}~\mathrm{for~all}~i\in \mathcal{B}\cap [N_{m-1}],\\
            g_{m}(x)\equiv &(w_i^{(m)}(x)+H_{m}(s_i(x)))\pmod{m_i(x)}~\mathrm{for~all}~i\in \mathcal{B}\cap [N_{m-1}],\\
            g_{m}(x)\equiv &s_i(x)\pmod{m_i(x)}~\mathrm{for~all}~i\in \mathcal{B}\cap [N_{m-1}+1, n].\\
           \end{aligned} \right.
        \end{displaymath}
Namely,
       \begin{displaymath}
         \left\{\begin{aligned}
            g_{1}(x)\equiv &g_{1}(x)\equiv\cdots\equiv g_{m}(x)\equiv s(x)\pmod{m_0(x)},\\
            g_{1}(x)\equiv &s_i^{(1)}(x)\pmod{m_i(x)}~\mathrm{for~all}~i\in \mathcal{B}\cap [N_1],\\
            g_{2}(x)\equiv &s_i^{(2)}(x)\pmod{m_i(x)}~\mathrm{for~all}~i\in \mathcal{B}\cap [N_2],\\
                     &\vdots\\
            g_{m-1}(x)\equiv &s_i^{(m-1)}(x)\pmod{m_i(x)}~\mathrm{for~all}~i\in \mathcal{B}\cap [N_{m-1}],\\
            g_{m}(x)\equiv &s_i^{(m)}(x)\pmod{m_i(x)}~\mathrm{for~all}~i\in \mathcal{B}\cap [N_m].\\
           \end{aligned} \right.
       \end{displaymath}
where
  \[s_i^{(\ell)}(x)=\begin{cases}
            H_{\ell}(s_i(x))+w_i^{(\ell)}(x),~\mathrm{if}~\ell\in [m-1], i\in \mathcal{B}\cap [N_{\ell}],\\
            H_{m}(s_i(x))+w_i^{(m)}(x),~\mathrm{if}~\ell=m, i\in \mathcal{B}\cap [N_{m-1}],\\
            s_i(x),~\mathrm{if}~\ell=m, i\in \mathcal{B}\cap [N_{m-1}+1,N_m].\\
         \end{cases}\]
From Lemma \ref{le: CRT} it holds that
    \[g_{\ell}(x)\equiv\sum\limits_{i\in \mathcal{\widetilde{B}}^{(\ell)}}\lambda_{i,\mathcal{\widetilde{B}}^{(\ell)}}(x)M_{i,\mathcal{\widetilde{B}}^{(\ell)}}(x)s_i^{(\ell)}(x)\pmod {M_{\mathcal{\widetilde{B}}^{(\ell)}}(x)},\]
for $\ell\in [m]$, where $\mathcal{\widetilde{B}}^{(\ell)}=\{0\}\cup(\mathcal{B}\cap [N_{\ell}])$, $s_0^{(\ell)}(x)=s(x)$,
$M_{\mathcal{\widetilde{B}}^{(\ell)}}(x)=\prod\limits_{i\in\mathcal{\widetilde{B}}^{(\ell)}}m_i(x)$, $M_{i,\mathcal{\widetilde{B}}^{(\ell)}}(x)=M_{\mathcal{\widetilde{B}}^{(\ell)}}(x)/m_i(x)$, and $\lambda_{i,\mathcal{\widetilde{B}}^{(\ell)}}(x)\equiv M_{i,\mathcal{\widetilde{B}}^{(\ell)}}^{-1}(x)\pmod {m_i(x)}$.

Denote by $g_{\ell,\mathcal{\widetilde{B}}^{(\ell)}}(x)=\sum\limits_{i\in \mathcal{\widetilde{B}}^{(\ell)}}\lambda_{i,\mathcal{\widetilde{B}}^{(\ell)}}(x)M_{i,\mathcal{\widetilde{B}}^{(\ell)}}(x)s_i^{(\ell)}(x)\pmod {M_{\mathcal{\widetilde{B}}^{(\ell)}}(x)}$, then
    \[\deg(g_{\ell,\mathcal{\widetilde{B}}^{(\ell)}}(x))<\deg(M_{\mathcal{\widetilde{B}}^{(\ell)}}(x))=\sum\limits_{i\in \mathcal{\widetilde{B}}^{(\ell)}}d_i,\]
and
      \begin{equation}\label{Eq: preimage F}
     g_{\ell}(x)\equiv g_{\ell,\mathcal{\widetilde{B}}^{(\ell)}}(x)\pmod{M_{\mathcal{\widetilde{B}}^{(\ell)}}(x)}.
     \end{equation}
Therefore, there is a polynomial $k_{\ell,\mathcal{\widetilde{B}}^{(\ell)}}(x)\in \mathbb{F}_p[x]$ such that
     \[g_{\ell}(x)=g_{\ell,\mathcal{\widetilde{B}}^{(\ell)}}(x)+k_{\ell,\mathcal{\widetilde{B}}^{(\ell)}}(x) M_{\mathcal{\widetilde{B}}^{(\ell)}}(x).\]

On the one hand, it holds that $\deg(g_{\ell}(x))<\sum\limits_{i=1}^{t_{\ell}}d_i$ from the definition of $\mathcal{F}$ in Expression (\ref{Eq: definiton F}). Since $\deg(g_{\ell,\mathcal{\widetilde{B}}^{(\ell)}}(x))<\deg(M_{\mathcal{\widetilde{B}}}(x))$, we get that
        \begin{equation}
        \deg(k_{\ell,\mathcal{\widetilde{B}}^{(\ell)}}(x))<\sum\limits_{i=1}^{t_\ell}d_i-\sum\limits_{i\in \mathcal{\widetilde{B}}^{(\ell)}}d_i.
        \end{equation}
This implies that there are $p^{\theta_{\ell}}$ different choices for $k_{\ell,\mathcal{\widetilde{B}}^{(\ell)}}(x)$, where
     \[\theta_{\ell}=\sum\limits_{i=1}^{t_\ell}d_i-\sum\limits_{i\in \mathcal{\widetilde{B}}^{(\ell)}}d_i=\sum\limits_{i=1}^{t_\ell}d_i-\sum\limits_{i\in \mathcal{B}\cap [N_{\ell}]}d_i-d_0\geq 0.\]
On the other hand, different choices for $k_{\ell,\mathcal{\widetilde{B}}^{(\ell)}}(x)\in \mathbb{F}_p[x]$ with $\deg(k_{\ell,\mathcal{\widetilde{B}}^{(\ell)}}(x))<\theta_{\ell}$ correspond to different $g_{\ell}(x)\in \mathbb{F}_p[x]$ satisfying $\deg(g_{\ell}(x))<\sum\limits_{i=1}^{t_{\ell}}d_i$ and Expression (\ref{Eq: preimage F}), i.e., $(g_1(x),g_2(x),\ldots,g_m(x))\in \mathcal{F}$. Consequently, the cardinality of $\Phi^{-1}(s(x))$ is
        \[|\Phi^{-1}(s(x))|=\prod_{\ell=1}^{m}p^{\theta_{\ell}}=p^{\sum_{i=1}^{m}\theta_{\ell}}=p^{\theta},\]
where $\theta=\sum_{\ell=1}^{m}\left(\sum\limits_{i=1}^{t_\ell}d_i-\sum\limits_{i\in \mathcal{B}\cap [N_{\ell}]}d_i-d_0\right)\geq 0$.
\end{proof}

\begin{lemma}\label{le: cardinality of F}
 The cardinality of $\mathcal{F}$ is $|\mathcal{F}|=p^{\theta+d_0}$.
\end{lemma}

\begin{proof}
From the proof of Lemma \ref{le: preimage}, we get that for any $(g_1(x),g_2(x),\ldots,g_m(x))\in \mathcal{F}$, the following system of congruences are satisfied.
          \begin{displaymath}
         \left\{\begin{aligned}
            g_{1}(x)\equiv &g_{1}(x)\equiv\cdots\equiv g_{m}(x)\pmod{m_0(x)},\\
            g_{1}(x)\equiv &s_i^{(1)}(x)\pmod{m_i(x)}~\mathrm{for~all}~i\in \mathcal{B}\cap [N_1],\\
            g_{2}(x)\equiv &s_i^{(2)}(x)\pmod{m_i(x)}~\mathrm{for~all}~i\in \mathcal{B}\cap [N_2],\\
                     &\vdots\\
            g_{m-1}(x)\equiv &s_i^{(m-1)}(x)\pmod{m_i(x)}~\mathrm{for~all}~i\in \mathcal{B}\cap [N_{m-1}],\\
            g_{m}(x)\equiv &s_i^{(m)}(x)\pmod{m_i(x)}~\mathrm{for~all}~i\in \mathcal{B}\cap [N_m].\\
           \end{aligned} \right.
       \end{displaymath}
where
  \[s_i^{(\ell)}(x)=\begin{cases}
            H_{\ell}(s_i(x))+w_i^{(\ell)}(x),~\mathrm{if}~\ell\in [m-1], i\in \mathcal{B}\cap [N_{\ell}],\\
            H_{m}(s_i(x))+w_i^{(m)}(x),~\mathrm{if}~\ell=m, i\in \mathcal{B}\cap [N_{m-1}],\\
            s_i(x),~\mathrm{if}~\ell=m, i\in \mathcal{B}\cap [N_{m-1}+1,N_m].\\
         \end{cases}\]

Assume that $s(x)\in \mathcal{S}=\{g(x)\in \mathbb{F}_p[x]:\deg(g(x))<d_0\}$. From Lemma \ref{le: preimage}, there are exactly $p^{\theta}$ different choices for $(g_1(x),g_2(x),\ldots,g_m(x))\in \mathcal{F}$ such that
            \[g_{1}(x)\equiv g_{1}(x)\equiv\cdots\equiv g_{m}(x)\equiv s(x)\pmod{m_0(x)}.\]
Moreover, different $s(x)\in \mathcal{S}$ corresponds different $(g_1(x),g_2(x),\ldots,g_m(x))\in \mathcal{F}$. Since the cardinality of $\mathcal{S}$ is $p^{d_0}$, then the cardinality of $\mathcal{F}$ is $|\mathcal{F}|=p^{d_0}\cdot p^{\theta}=p^{\theta+d_0}$.
\end{proof}

\begin{theorem}[Asymptotic perfectness]\label{Th: our perfectness}
   Our scheme is asymptotically perfect.
\end{theorem}

\begin{proof}
For any $s(x)\in \mathcal{S}$ and $\mathcal{B}\notin \Gamma$, by Lemmas \ref{le: cardinality of F} and \ref{le: preimage} we get the conditional probability
\[\mathsf{Pr}(\mathbf{S}=s(x)|\mathbf{V}^\prime_\mathcal{B}=\mathcal{V}^\prime_{\mathcal{B}})=\frac{|\Phi^{-1}(s(x))|}{|\mathcal{F}|}=\frac{p^{\theta}}{p^{\theta+d_0}}=\frac{1}{p^{d_0}}.\]
Therefore,
\begin{displaymath}
             \begin{aligned}
              &\mathsf{H}(\textbf{S}|\textbf{V}^\prime_\mathcal{B})\\
                    =&-\sum_{\mbox{\tiny$\begin{array}{c}
                                       \mathcal{B}\subset[n],\\
                                       \mathcal{B}\notin \Gamma\\
                                       \end{array}$}}
                      \sum_{s(x)\in \mathcal{S}}\mathsf{Pr}(\textbf{V}^\prime_\mathcal{B}=\mathcal{V}^\prime_{\mathcal{B}})
                      \mathsf{Pr}(\mathbf{S}=s(x)|\textbf{V}^\prime_\mathcal{B}=\mathcal{V}^\prime_{\mathcal{B}})\log_2 \mathsf{Pr}(\mathbf{S}=s(x)|\textbf{V}^\prime_\mathcal{B}=\mathcal{V}^\prime_{\mathcal{B}})\\
                    =&\sum_{\mbox{\tiny$\begin{array}{c}
                                       \mathcal{B}\subset[n],\\
                                       \mathcal{B}\notin\Gamma\\
                                       \end{array}$}}
                      \sum_{s(x)\in \mathcal{S}}\mathsf{Pr}(\textbf{V}^\prime_\mathcal{B}=\mathcal{V}^\prime_{\mathcal{B}})\frac{1}{p^{d_0}}\log_2 p^{d_0}\\
                     =&\sum_{\mbox{\tiny$\begin{array}{c}
                                       \mathcal{B}\subset[n],\\
                                       \mathcal{B}\notin\Gamma\\
                                       \end{array}$}}
                      \mathsf{Pr}(\textbf{V}^\prime_\mathcal{B}=\mathcal{V}^\prime_{\mathcal{B}})\log_2 p^{d_0}\\
                    =&\log_2 p^{d_0}.
            \end{aligned}
          \end{displaymath}
Moreover, the secret in our scheme is random and uniform, which means that $\mathsf{H}(\textbf{S})=\log_2|\mathcal{S}|=\log_2 p^{d_0}= \mathsf{H}(\mathbf{S}|\mathbf{V}^\prime_\mathcal{B})$ for any $\mathcal{B}\notin \Gamma$. According to Lemma \ref{le: loss entropy}, for all $\epsilon_1>0$, there is a positive integer $\sigma_1$ such that for all $|\mathcal{S}|=p^{d_0}>\sigma_1$, it holds that
                               \[0<\mathsf{H}(\mathbf{S})-\mathsf{H}(\mathbf{S}|\mathbf{V}_{\mathcal{B}})
                               =\mathsf{H}(\mathbf{S}|\mathbf{V}'_{\mathcal{B}})-\mathsf{H}(\mathbf{S}|\mathbf{V}_{\mathcal{B}})<\epsilon_1.\]
As a result, our scheme is asymptotically perfect by Definition \ref{def: Asymptotically Ideal DHSS}.
\end{proof}

\begin{theorem}\label{Theorem: summary}
Let $\mathcal{P}$ be a set of $n$ participants, and it is partitioned into $m$ disjoint subsets $\mathcal{P}_1, \mathcal{P}_2,\ldots, \mathcal{P}_m$. For a threshold sequence $t_1, t_2, \ldots, t_m$ such that $1\leq t_1< t_2<\cdots<t_m\leq n$ and $t_{\ell}\leq |\mathcal{P}_{\ell}|$ for $\ell\in [m]$, our scheme is a secure and asymptotically perfect DHSS scheme. Furthermore, our scheme is an asymptotically ideal DHSS scheme with an information one when $d_0=d_1=\cdots=d_n$.
\end{theorem}
\begin{proof}
According to Definition \ref{def: Multilevel secret sharing}, Theorem \ref{Theorem:Correctness of our scheme} and Theorem \ref{Th: our perfectness}, we get that our scheme is a secure and asymptotically perfect DHSS scheme. Moreover, it is known that
         \[s_i(x)=\begin{cases}
            c_i(x),\mathrm{if}~i\in [N_{m-1}],\\
            f_m(x) \pmod {m_i(x)},\mathrm{if}~i\in [N_{m-1}+1,N_m],\\
         \end{cases}\]
where $c_{i}(x)=c_{i,0}+c_{i,1}x+\cdots+c_{i,d_i-1}x^{d_i-1}\in \mathbb{F}_p[x], i\in [N_{m-1}]$, $f_{m}(x)=s(x)+\alpha_{m}(x)x^{d_0}$, and
$\alpha_{m}(x)\in \mathcal{G}_{m} =\{g(x)\in \mathbb{F}_p[x]:\deg(g(x))<(\sum\limits_{i=1}^{t_{m}}d_i)-d_0\}$. When $d_0=d_1=\cdots=d_n$, it is easy to check that every share space is the same as the secret space $\mathcal{S}=\{g(x)\in \mathbb{F}_p[x]:\deg(g(x))<d_0\}$, which implies that the information rate is $\rho=1$. According to Definition \ref{def: Asymptotically Ideal DHSS}, our scheme is an asymptotically ideal DHSS scheme with an information one when $d_0=d_1=\cdots=d_n$.
\end{proof}

Compared with the unconditionally secure and ideal DHSS schemes of \cite{Brickell1989,Tassa2007,Chenqi2022}, our scheme is a computationally secure and asymptotically ideal DHSS scheme, but it permits distinct share sizes among participants. Specifically, let $d_0=d_1=\cdots=d_n$, then our scheme's share size is $d_0\log_2 p$ bits for $d_0\geq 1$ and $p>n$, which is smaller than that of \cite{Brickell1989,Tassa2007,Chenqi2022} (see Table 1). Compared with the DHSS schemes of \cite{Harn-Miao2014,Oguzhan-Ersoy2016,Tiplea2021,Yangjing2024}, our scheme is secure and asymptotically ideal at the same time.

\section{Conclusion}\label{Sec: conclusion}
In this work, we give an attack method on the ideal DHSS scheme of Yang et al. \cite{Yangjing2024} and the asymptotically ideal DHSS scheme of Tiplea et al. \cite{Tiplea2021}. We further propose a novel asymptotically ideal DHSS scheme. Compared with previous CRT-based DHSS scheme, our scheme is a secure and asymptotically ideal DHSS scheme with an information rate one. How to reduce the number of public values is our future work.


\begin{thebibliography}{99}

\bibitem{Shamir1979} A. Shamir, How to share a secret, Communications of the ACM, vol.22, no.11, pp.612-613, 1979.

\bibitem{Blakley1979} G. R. Blakley, Safeguarding cryptographic keys, in: 1979 International Workshop on Managing Requirements Knowledge, MARK, New York, NY, USA, pp. 313-318. IEEE, 1979.

\bibitem{Simmons} G. J. Simmons, How to (really) share a secret, in: Advances in Cryptology-Crypto'88, Santa Barbara, California, USA. Lecture Notes in Computer Science, vol.403, pp.390-448. Springer, 1988.

\bibitem{Brickell1989} E. F. Brickell, Some ideal secret sharing schemes, in: Advances in Cryptology-Eurocrypt'89, Houthalen, Belgium. Lecture Notes in Computer Science, vol.434, pp.468-475. Springer, 1989.

\bibitem{Tassa2007} T. Tassa, Hierarchical threshold secret sharing, Journal of Cryptology, vol.20, no.2, pp.237-264, 2007.

\bibitem{Chenqi2022} Q. Chen, C. Tang, and Z. Lin, Efficient explicit constructions of multipartite secret sharing schemes, IEEE Transactions on Information Theory, vol.68, no.1, pp.601-631, 2022.

\bibitem{Harn-Miao2014} L. Harn, and F. MIAO, Multilevel threshold secret sharing based on the Chinese remainder theorem, Information Processing Letters, vol.114, no.9, 2014. 504-509.

\bibitem{Oguzhan-Ersoy2016} O. Ersoy, K. Kaya, and K. Kaskaloglu, Multilevel Threshold Secret and Function Sharing based on the Chinese Remainder Theorem, arXiv: 1605.07988, https://arxiv.org/abs/1605.07988.

\bibitem{Tiplea2021}F. L. Tiplea, and C. C. Dragan, Asymptotically ideal Chinese remainder theorem-based secret sharing schemes for multilevel and compartmented
access structures, IET Information Security, vol. 15, no. 4, pp. 282-296, 2021.

\bibitem{Yangjing2024} J. Yang, S.-T. Xia, X. Wang, J. Yuan, and F.-W. Fu, A perfect ideal hierarchical secret sharing scheme based on the CRT for Polynomial Rings, in: 2024 IEEE International Symposium on Information Theory (ISIT), Athens, Greece, pp.321-326. IEEE, 2024.

\bibitem{Ning2018} Y. Ning, F. Miao, W. Huang, K. Meng, Y. Xiong, and X. Wang, Constructing ideal secret sharing schemes based on Chinese Remainder Theorem, in: Advances in Cryptology - {ASIACRYPT} 2018, Brisbane, QLD, Australia. Lecture Notes in Computer Science, vol. 11274, pp. 310-331. Springer, 2018.

\bibitem{Ding2023} J. Ding, P. Ke, C. Lin, and H. Wang, Ramp scheme based on CRT for polynomial ring over finite field, Journal of Systems Science and Complexity, vol. 36, no.1, pp. 129-150, 2023.

\bibitem{RM1985} R. C. Merkle. One Way Hash Functions and DES. Advances in Cryptology,CRYPTO'89, LNCS435,PP. 428-466,1985．

\bibitem{Quisquater2002} M. Quisquater, B. Preneel, and J. Vandewalle, On the security of the threshold scheme based on the Chinese remainder theorem, in: Public Key Cryptography, PKC 2002. Lecture Notes in Computer Science, vol. 2274, pp.199-210. Springer, Berlin, Heidelberg, 2002.





























\end{thebibliography}
\end{document}